\documentclass[preprint]{elsarticle}
  \usepackage{amssymb,amsmath}
  \usepackage{enumerate}
  \usepackage{graphicx}
  \usepackage{verbatim}
  \usepackage[noend]{algorithmic}
  \usepackage[boxed]{algorithm}
  \usepackage{diagbox}
  \usepackage{tikz}
  
  \floatname{algorithm}{Algorithm}

  \usepackage[algo2e, noend, noline, boxed]{algorithm2e}
  \SetKwComment{tcm}{\{}{\}}
   
  \newenvironment{myalgorithm}[2][htbp]
  {%
    \setlength{\algomargin}{.2cm}
    \begin{center}
    \begin{minipage}{#2}
    \begin{algorithm2e}[#1]
    \small
     \let\Par=\par
       \def\par{\endgraf\vspace{.1cm}}
           \SetKw{To}{to}%
       \SetKw{Downto}{downto}%
           \SetKw{Or}{or}%
       \SetKwFor{Algo}{Algorithm}{}{}%
      \vspace{.15cm}%
   }
   {%
     \let\par=\Par\end{algorithm2e}%
     \end{minipage}%
     \end{center}%
   }

  \newtheorem{theorem}{Theorem}
  \newtheorem{lemma}[theorem]{Lemma}
  \newtheorem{example}{Example}
  \newproof{proof}{Proof}

  \newcommand{\TT}{\mathcal{T}}
  \newcommand{\lccp}{\mathit{lccp}}
  \newcommand{\lcp}{\mathit{lcp}}
  \newcommand{\LCCP}{\mathit{LCCP}}
  \newcommand{\solid}{\mathit{solid}}
  \newcommand{\type}{\mathit{type}}

  \newcommand{\NextChange}{\mathit{NextChange}}
  \newcommand{\Next}{\mathit{next}}
  
  \newcommand{\SUF}{\mathit{SUF}}
  \newcommand{\LCP}{\mathit{LCP}}
  \newcommand{\RANK}{\mathit{RANK}}
  
  \def\rdots{\mathinner{\ldotp\ldotp}}

\begin{document}

\begin{frontmatter}

  \title{
      A Note on the Longest Common Compatible Prefix Problem for Partial Words
  }

  \author[kcl,upe]{M.~Crochemore}
  \ead{maxime.crochemore@kcl.ac.uk}

  \author[kcl,cut]{C.~S.~Iliopoulos}
  \ead{c.iliopoulos@kcl.ac.uk}

  \author[uw]{T.~Kociumaka}
  \ead{kociumaka@mimuw.edu.pl}

  \author[uw]{M.~Kubica}
  \ead{kubica@mimuw.edu.pl}

  \author[kcl]{A.~Langiu}
  \ead{alessio.langiu@kcl.ac.uk}

  \author[uw]{J.~Radoszewski\corref{cor1}}
  \ead{jrad@mimuw.edu.pl}

  \author[uw,umk]{W.~Rytter%\fnref{grytter}
  }
  \ead{rytter@mimuw.edu.pl}

  \author[uw]{B.~Szreder}
  \ead{szreder@mimuw.edu.pl}

  \author[bio,uw]{T.~Wale\'n}
  \ead{walen@mimuw.edu.pl}

  \cortext[cor1]{Corresponding author. Tel.: +48-22-55-44-484,
    fax: +48-22-55-44-400.
  }

%  \fntext[grytter]{The author is supported by grant no.\ N206 566740 of the National Science Centre.}

  \address[kcl]{King's College London, London WC2R 2LS, UK}

  \address[upe]{Universit\'e Paris-Est, France}

  \address[cut]{Digital Ecosystems \& Business Intelligence Institute,
    Curtin University of Technology, Perth WA 6845, Australia}

  \address[uw]{Faculty~of Mathematics, Informatics and Mechanics,
    University of Warsaw, ul.~Banacha~2, 02-097 Warsaw, Poland}

  \address[umk]{
    Faculty of Mathematics and Computer Science, Copernicus University,\\
    ul. Chopina 12/18, 87-100~Toru\'n, Poland
  }

  \address[bio]{Laboratory of Bioinformatics and Protein Engineering,
     International Institute of Molecular and Cell Biology in Warsaw, Poland}

  \begin{abstract}
    For a partial word $w$ the longest common compatible prefix of two positions $i,j$, denoted $\lccp(i,j)$,
    is the largest $k$ such that $w[i,i+k-1]\uparrow w[j,j+k-1]$, where $\uparrow$ is the compatibility relation of partial words
    (it is not an equivalence relation).
    The LCCP problem is to preprocess a partial word in such a way that any query $\lccp(i,j)$ about this word
    can be answered in $O(1)$ time.
    It is a natural generalization of the longest common prefix (LCP) problem for
    regular words, for which an $O(n)$ preprocessing time and $O(1)$ query time solution exists.

    Recently an efficient algorithm for this problem has been
    given by F.~Blanchet-Sadri and J.~Lazarow (LATA 2013).
    The preprocessing time was $O(nh+n)$, where $h$ is the number of ``holes'' in $w$.
    The algorithm was designed for partial words over a constant alphabet
    and was quite involved.

    We present a simple solution to this problem with slightly better runtime
    that works for any linearly-sortable alphabet.
    Our preprocessing is in time $O(n\mu+n)$, where $\mu$ is the number
    of blocks of holes in $w$.
    Our algorithm uses ideas from alignment algorithms and dynamic programming.
  \end{abstract}


\end{frontmatter}

  %\section{Introduction}
  Let $w$ be a partial word of length $n$.
  That is, $w=w_1 \ldots w_n$, with $w_i \in \Sigma \cup \{\diamond\}$, where
  $\Sigma$ is called the alphabet (the set of letters) and $\diamond \notin \Sigma$ denotes a hole.
  A non-hole position in $w$ is called solid.
  By $h$ we denote the number of holes in $w$ and by $\mu$ we denote
  the number of blocks of consecutive holes in $w$.

  By $\uparrow$ we denote the compatibility relation: $a \uparrow \diamond$ for any $a \in \Sigma$
  and moreover $\uparrow$ is reflexive.
  The relation $\uparrow$ is extended in a natural letter-by-letter manner to partial words of the same length.
  Note that $\uparrow$ is not transitive: $a \uparrow \diamond$ and $\diamond \uparrow b$
  whereas $a \not\uparrow b$ for any letters $a \ne b$.

  Motivation on partial words and their applications can be found in the book~\cite{PartialWords}.

  \begin{example}\label{ex:1}
    Let $w = a\; b \diamond \diamond\; a \diamond \diamond \diamond b\; c\; a\; b\; \diamond$.
    There are 7 solid positions in $w$, $h=6$ and $\mu=3$.
  \end{example}

  By $w[i,j]$ we denote the subword $w_i \ldots w_j$.
  The longest common compatible prefix of two positions $i,j$, denoted $\lccp(i,j)$,
  is the largest $k$ such that $w[i,i+k-1]\uparrow w[j,j+k-1]$.

  \begin{example}
    For the word $w$ from Example~\ref{ex:1}, we have
    $\lccp(2,9)=3$,
    $\lccp(1,2)=0$,
    $\lccp(3,6)=8$.
  \end{example}

  In \cite{DBLP:conf/lata/Blanchet-SadriL13} F.~Blanchet-Sadri and J.~Lazarow
  provide a data structure that is constructed in $O(nh+n)$ time and space
  and allows computing LCCP for any two positions in $O(1)$ time.
  Their data structure is based on suffix dags which are a modification of suffix trees
  and requires $\Sigma$ to be a fixed alphabet (i.e.\ $|\Sigma|=O(1)$).

  We show a much simpler data structure that requires only $O(n\mu+n)$ construction
  time and space and also allows constant-time LCCP-queries.
  Our algorithm is based on alignment techniques and suffix arrays for full (regular) words
  and works for any integer alphabet (that is, the letters can be treated as integers
  in a range of size $n^{O(1)}$).

  %\section{Main algorithm}
  By $\type(i)$ we mean \emph{hole} or \emph{solid} depending on the type of $w_i$.
  We add a sentinel position: $w_0 = \diamond$ if $w_1$ is solid or $w_0=a \in \Sigma$ if $w_1$ is a hole.
  A position in $w$ is called \emph{transit} if it is a hole directly preceded by a solid position
  or a solid position directly preceded by a hole, see Fig.~\ref{fig:transit}.
  Enumerate all transit positions $\TT = \{i_1,i_2,\ldots,i_{\kappa}\}$.
  Note that $\kappa \le 2 \mu$.

  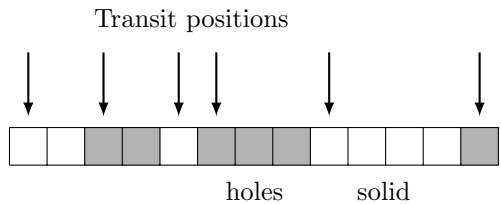
\begin{figure}[htpb]
  \begin{center}
    \begin{tikzpicture}[scale=0.5]
      \filldraw[white!70!black] (2,0) rectangle (4,1);
      \filldraw[white!70!black] (5,0) rectangle (8,1);
      \filldraw[white!70!black] (12,0) rectangle (13,1);
      \draw (0,0) grid (13,1);
      \draw[->,-latex,thick] (0.5,3) -- (0.5,1.3);
      \draw[xshift=2cm,->,-latex,thick] (0.5,3) -- (0.5,1.3);
      \draw[xshift=4cm,->,-latex,thick] (0.5,3) -- (0.5,1.3);
      \draw[xshift=5cm,->,-latex,thick] (0.5,3) -- (0.5,1.3);
      \draw[xshift=8cm,->,-latex,thick] (0.5,3) -- (0.5,1.3);
      \draw[xshift=12cm,->,-latex,thick] (0.5,3) -- (0.5,1.3);
      \draw (2,3.3) node[above right] {Transit positions};
      \draw (5.5,-0.2) node[below right] {holes};
      \draw (9,-0.2) node[below right] {solid};
    \end{tikzpicture}
  \end{center}
    \caption{\label{fig:transit}
      Illustration of transit positions, $\mu=3$, $\kappa=6$.
    }
  \end{figure}

  \begin{example}
    Let $w = a\; b \diamond \diamond\;a \diamond \diamond \diamond b\; c\; a\; b\;\diamond$.
    Then $\TT = \{1,3,5,6,9,13\}$, see also Fig.~\ref{fig:transit}.
  \end{example}

  \noindent
  Our data structure consists of two parts:

  \begin{enumerate}
    \item
    a data structure of size $O(n)$ allowing to answer in $O(1)$ time
    the {\em longest common prefix},
    denoted $\lcp(i,j)$, between any two positions in the full word
    $\hat{w}$, which results from $w$ by treating holes as solid symbols

    \item
    a $n \times \mu$ table
    $$\LCCP[i,j]\;=\; \lccp(i,j)\quad \mbox{for }i \in \{1,\ldots,n\},\ j \in \TT.$$
    For simplicity we assume $\LCCP[j,i] = \LCCP[i,j]$ if $i \in \TT$, $j \in \{1,\ldots,n\}$.
  \end{enumerate}

  The data structure (1) consists of the suffix array for $\hat{w}$ and Range Minimum Query data structure.
  A suffix array is composed of three tables: $\SUF$, $\RANK$ and $\LCP$.
  The $\SUF$ table stores the list of positions in $\hat{w}$ sorted according to
  the increasing lexicographic order of suffixes starting at these positions.
  The $\LCP$ array stores the lengths of the longest common prefixes
  of consecutive suffixes in $\SUF$.
  We have $\LCP[1] = -1$ and, for $1 < i \le n$, we have:
  $$\LCP[i] = \lcp(\SUF[i - 1], \SUF[i]).$$
  Finally, the $\RANK$ table is an inverse of the $\SUF$ table:
  $$\SUF[\RANK[i]] = i \quad\mbox{for }i=1,2,\ldots,n.$$
  All tables comprising the suffix array for a word over a linearly-sortable alphabet
  can be constructed in $O(n)$ time
  \cite{AlgorithmsOnStrings,DBLP:conf/icalp/KarkkainenS03,DBLP:conf/cpm/KasaiLAAP01}.

  The Range Minimum Query data structure (RMQ, in short) is constructed for
  an array $A[1 \rdots n]$ of integers.
  This array is preprocessed to answer the following form of queries: for an
  interval $[i, j]$ (where $1 \le i \le j \le n$), find the minimum value
  $A[k]$ for $i \le  k \le j$.
  The best known RMQ data structures have $O(n)$ preprocessing time and
  $O(1)$ query time \cite{DBLP:journals/siamcomp/HarelT84}.

  To compute $\lcp(i,j)$ for $i \ne j$ we use a classic combination of the two data structures, see also \cite{AlgorithmsOnStrings}.
  Let $x$ be $\min(\RANK[i],\RANK[j])$ and $y$ be $\max(\RANK[i],\RANK[j])$.
  Then:
  $$\lcp(i,j) = \min\{\LCP[x+1], \LCP[x+2],\ldots,\LCP[y]\}.$$
  This value can be computed in $O(1)$ time provided that RMQ data structure for the table $\LCP$ is given.

  \bigskip
  For $i \in \{1,\ldots,n\}$ define
  $$\NextChange[i] = \min\{k > 0\,:\, \type(i+k) \ne \type(i)\}.$$
  If no such $k$ exists then $\NextChange[i] = n+1-i$.
  Clearly the $\NextChange$ table can be computed in $O(n)$ time.
  We denote
  $$\Next(i,j) = \min(\NextChange[i],\NextChange[j]).$$

  \begin{lemma}\label{lem:lccp-query}
    Assume we have the data structures from points (1)-(2) above.
    Then $\lccp(i,j)$ for any $1 \le i,j \le n$ can computed in $O(1)$ time.
  \end{lemma}
  \begin{proof}
    If any of the positions $i,j$ belongs to $\TT$ then we simply use the $\LCCP$ table.
    Otherwise we have two cases.

    If any of the positions $i,j$ is a hole then the result is $d+\LCCP[i+d,j+d]$, where $d=\Next(i,j)$.

    Otherwise, both $i,j$ are solid.
    Let $k=\lcp(i,j)$.
    The result is $d+\LCCP[i+d,j+d]$ if $k \ge d$ or $k$ otherwise.
    \qed
  \end{proof}
  
  \begin{myalgorithm}[H]{9.2 cm}
    \Algo{\textsc{LCCP-Query}($w,i,j$)}
    {
%      \lIf{$(i,j)\in \TT$}{
%        \Return{$\LCCP[i,j]$}\;
%      }
      $d:=\Next(i,j)$; $k:=\lcp(i,j)$\;
      \If{$\type(w_i)\not=\solid$ {\bf or} $\type(w_j)\not=\solid$ {\bf or} $k\ge d$}{
          \Return{$d+\LCCP[i+d,j+d]$}\;
      }\lElse{
          \Return{$k$}
      }
    }
  \end{myalgorithm}

  \begin{theorem}
    Let $w$ be a partial word of length $n$ over an integer alphabet.
    We can preprocess $w$ in $O(n\mu+n)$ time to enable $\lccp$-queries
    in constant time.
  \end{theorem}
  \begin{proof}
    The data structure (1) for $\lcp$-queries is constructed in $O(n)$ time
    from the suffix array for $\hat{w}$ and the RMQ data structure for the $\LCP$ table.
    The construction of the data structure (2) is shown in the following \mbox{\textsc{LCCP-Preprocess}} algorithm.
    This algorithm is based on the dynamic programming technique and works in $O(n\mu+n)$ time.
    Using the two data structures, by Lemma~\ref{lem:lccp-query} we can answer $\lccp(i,j)$ queries
    in $O(1)$ time.
    \qed
  \end{proof}

  \newcommand{\algname}{\textsc{LCCP-Preprocess}}
  \begin{myalgorithm}[H]{10.7 cm}
    \Algo{\algname(w)}
    {
      \For{$i:=1$ \mbox{\bf to} $n+1$}{
        $\LCCP[i,n+1] := 0$\;
      }
      \ForEach{$(i,j)$ : $i \in \{1,\ldots,n\},\ j \in \TT$ in decreasing lex.\ order} {
        $\LCCP[i,j]:=$\,\textsc{LCCP-Query}($w,i,j$);
      }
    }
  \end{myalgorithm}

  \begin{example}
    Let $w = a\; b \diamond \diamond\; a \diamond \diamond \diamond b\; c\; a\; b\;\diamond$.
    The $\LCCP$ table computed by the algorithm \textsc{LCCP-Preprocess}($w$) is as follows.

    \vspace{0.4cm}
    \renewcommand\arraystretch{1.2}
    \begin{tabular}{c|ccccccccccccc}
      \diagbox{$j$}{$w_i$} & $a$ & $b$ & $\diamond$ & $\diamond$ & $a$ & $\diamond$ & $\diamond$ &   $\diamond$ & $b$ & $c$ & $a$ & $b$ & $\diamond$\\
      \hline \hline
      1 & 13 & 0 & 8 & 1 & 4 & 4 & 7 & 4 & 0 & 0 & 3 & 0 & 1\\
      3 & 8 & 7 & 11 & 6 & 6 & 8 & 2 & 2 & 5 & 2 & 3 & 2 & 1\\
      5 & 4 & 0 & 6 & 5 & 9 & 4 & 4 & 6 & 0 & 0 & 3 & 0 & 1\\
      6 & 4 & 3 & 8 & 5 & 4 & 8 & 3 & 3 & 5 & 4 & 3 & 2 & 1\\
      9 & 0 & 3 & 5 & 1 & 0 & 5 & 2 & 1 & 5 & 0 & 0 & 2 & 1\\
      13 & 1 & 1 & 1 & 1 & 1 & 1 & 1 & 1 & 1 & 1 & 1 & 1 & 1\\
    \end{tabular}
    \renewcommand\arraystretch{1.0}
  \end{example}

  \begin{comment}
        \section{Suffix automata of partial words}
        For a partial word $w$ let $Compat(w)$ be the set of all solid words
        compatible with $w$, and let
        $SufCompat(w)$ be the set of all solid suffixes of words in $Compat(w)$
        (in other words it is the set
        of all solid words compatible with partial suffixes of $w$).

        Denote by $SufAut(w)$ a smallest deterministic automaton accepting
        $SufCompat(w)$ (all such automata for the same $w$ are isomorphic).

        \begin{theorem}{\rm \bf [Wishful thinking]}\\
          The size of $SufAut(w)$ is $O(n\cdot \mu)$.
        \end{theorem}
  \end{comment}

  \bibliographystyle{abbrv}
  \bibliography{lccp}

\end{document}